\documentclass{article}

  \usepackage[preprint]{neurips_2026}

\usepackage[utf8]{inputenc} 
\usepackage[T1]{fontenc}    
\usepackage{hyperref}       
\usepackage{url}            
\usepackage{booktabs}       
\usepackage{amsfonts}       
\usepackage{nicefrac}       
\usepackage{microtype}      
\usepackage{xcolor}         

\usepackage{amsmath}
\usepackage{amssymb}
\usepackage{stmaryrd}
\usepackage{amsthm}

\newtheorem{theorem}{Theorem}
\newtheorem{lemma}{Lemma}

\title{Differential Privacy in the Extensive-Form Bandit Problem}

\author{%
  Stephen Pasteris \\
  The Alan Turing Institute\\
  London, United Kingdom \\
  \texttt{spasteris@turing.ac.uk} \\
   \And
  Rahul Savani \\
  The Alan Turing Institute;\\
  The University of Liverpool \\
  Liverpool, United Kingdom \\
  \texttt{rahul.savani@liverpool.ac.uk} \\
   \AND
  Theodore Turocy \\
  The Alan Turing Institute;\\
  The University of East Anglia \\
  Norwich, United Kingdom \\
   \texttt{t.turocy@uea.ac.uk} \\
}

\begin{document}

\maketitle

\begin{abstract}
We consider the extensive-form bandit problem, where on each trial the learner (a user coordinated by a server) plays an extensive-form game against an oblivious adversary, observing the information sets it finds itself in as well as the resulting payoff/loss. We give an algorithm for this problem that satisfies $\epsilon$-local differential privacy and attains a regret of $\tilde{O}(\sqrt{A\ln(S)T}/\epsilon)$, where $A$ is the total number of actions that the learner can possibly take, $S$ is the number of the learner's possible reduced strategies, and $T$ is the number of trials. On each trial, the time complexity of our algorithm is, up to a factor logarithmic in the maximum number of actions at an infoset, equal to the time required for the server to transmit the reduced strategy to the user. We note that local differential privacy is the strongest version of differential privacy and, to the best of our knowledge, this is the first work to study differential privacy of any form in the extensive-form bandit problem.
\end{abstract}

\newcommand{\nc}[1]{\newcommand{#1}}
\nc{\be}{\begin{equation*}}
\nc{\ee}{\end{equation*}}
\nc{\lss}{\mathcal{S}}
\nc{\nss}{\mathcal{S}'}
\nc{\bns}{\boldsymbol{\sigma}'}
\nc{\bls}{\boldsymbol{\sigma}}
\nc{\lst}[1]{\sigma_{#1}}
\nc{\nst}[1]{\sigma'_{#1}}
\nc{\lost}[1]{\ell_{#1}}
\nc{\ls}{\sigma}
\nc{\ns}{\sigma'}
\nc{\loss}[2]{\Lambda(#1,#2)}
\nc{\nma}{A}
\nc{\reg}{R}
\nc{\expt}[1]{\mathbb{E}\left[#1\right]}
\nc{\ep}{\epsilon}
\nc{\cnbs}{\boldsymbol{\mu}'}
\nc{\cnbsh}{\hat{\boldsymbol{\mu}}'}
\nc{\cnst}[1]{\mu'_{#1}}
\nc{\cnsht}[1]{\hat{\mu}'_{#1}}
\nc{\clbs}{\boldsymbol{\mu}}
\nc{\cprob}[2]{\mathbb{P}\left[#1\,|\,#2\right]}
\nc{\alg}{\textsc{DP-EFB}}
\nc{\indi}[1]{\llbracket #1\rrbracket}
\nc{\ver}{\mathcal{V}}
\nc{\lns}{\mathcal{N}}
\nc{\nns}{\mathcal{A}}
\nc{\tns}{\mathcal{L}}
\nc{\losf}{\lambda}
\nc{\anod}{v}
\nc{\ch}[1]{\mathcal{C}(#1)}
\nc{\rot}{r}
\nc{\us}{\sigma'}
\nc{\st}{\sigma}
\nc{\rf}{\mu}
\nc{\sts}{\mathcal{S}^\dag}
\nc{\rfs}{\mathcal{M}}
\nc{\rable}[1]{\mathcal{V}(#1)}
\nc{\ran}[1]{\mathcal{N}(#1)}
\nc{\rst}{\tilde{\sigma}}
\nc{\rsts}{\mathcal{S}}
\nc{\acs}{\mathcal{A}}
\nc{\sacs}[1]{\mathcal{A}(#1)}
\nc{\pnt}[1]{p(#1)}
\nc{\ac}{a}
\nc{\vth}{\vartheta_{t}}
\nc{\zti}[1]{\tilde{z}_{#1}}
\nc{\nosf}[1]{\phi_{#1}}
\nc{\cexpth}[2]{\mathbb{E}\left[#1\,\Bigg|\,#2\right]}
\nc{\rft}[1]{\mu_{#1}}
\nc{\stt}[1]{\sigma_{#1}}
\nc{\data}[1]{d_{#1}}
\nc{\datas}{\mathcal{D}}
\nc{\datass}{\mathcal{X}}
\nc{\sst}{\sigma}
\nc{\srf}{\mu}
\nc{\srfh}{\mu'}

\section{Introduction}

We consider a (perfect recall) extensive-form game. The \emph{learner} (which will be a user coordinated by a server) is one of the players of the game. Let $\rsts$ be the set of the learner's reduced strategies (i.e. those in the reduced strategic form of the game). We call any realisation of how the learner's opponents (including a potential ``chance'' player) play, an \emph{environment}. Given some $\st\in\rsts$ and an environment $\rf$, let $\loss{\st}{\rf}$ be the loss (the negation of the payoff) incurred by the learner when it plays $\st$ against $\rf$. Let $\nma$ be the total number of actions available to the learner. i.e. $\nma$ is the sum of the number of actions available at each information set (a.k.a. \emph{infoset}) of the learner.

We consider the following problem. We have $T$ trials where on each trial $t$ we have a user with an unknown environment $\rft{t}$ fixed a-priori. The way the users play the game will be coordinated by a server. On each trial $t\in[T]$ in turn:
\begin{itemize}
\item The server stochastically selects some $\stt{t}\in\rsts$ and sends it to the user.
\item The user plays the game against $\rft{t}$ using the reduced strategy $\stt{t}$, observing the infosets that it finds itself in as well as the resulting loss:
\be
\lost{t}:=\loss{\stt{t}}{\rft{t}}\,.
\ee
\item The user (based on what it has observed) stochastically selects some data $\data{t}$ from some space $\datas$ prescribed by the algorithm. The user then sends $\data{t}$ to the server.
\end{itemize}
Our performance is measured by the \emph{(expected) regret}, defined as:
\be
\reg := \expt{\sum_{t\in[T]}\lost{t}}-\min_{\st\in\rsts}\sum_{t\in[T]}\loss{\st}{\rft{t}}\ ,
\ee
which is the difference between our expected cumulative loss and that which would have been incurred by always playing the best fixed reduced strategy.

When there are no constraints on the data $\data{t}$ that is sent to the server on trial $t$, the problem has been well studied \cite{Kozuno2021ModelFreeLF, Bai2022NearOptimalLO, Fiegel2022AdaptingTG}. In particular, \cite{Fiegel2022AdaptingTG} achieves a regret of $\reg\in\tilde{\mathcal{O}}(\sqrt{H \nma T})$ where $H$ is the height of the learners tree of infosets.

In this paper we focus on maintaining $\ep$-\emph{local differential privacy} \cite{Kasiviswanathan2008WhatCW}. Essentially, this means that, on any trial $t$, one cannot deduce anything about $\rft{t}$ from the data $\data{t}$. Formally, we have some $\ep\in(0,1)$ and for any measurable $\datass\subseteq\datas$, any $\sst\in\rsts$ and any two environments $\srf$ and $\srfh$ we must have that:
\be
\cprob{\data{t}\in\datass}{\stt{t}=\sst, \rft{t}=\srf} \leq\exp(\ep)\cprob{\data{t}\in\datass}{\stt{t}=\sst, \rft{t}=\srfh}\ ,
\ee
where the probabilities are over the stochasticity of the user's selection of $\data{t}$. We note that local differential privacy is the strongest version of differential privacy. To the best of our knowledge, even the weaker notion of non-local differential privacy (where information is not kept private from the server but only the other users) is currently unexplored for this problem.

We will give an algorithm \alg\ that maintains $\ep$-local differential privacy whilst achieving a regret bound of:
\be
\reg\in\tilde{O}(\sqrt{A\ln(|\rsts|)T}/\ep)\,.
\ee
A salient feature of \alg\ is its extreme computational efficiency - up to a factor that is only logarithmic in the maximum number of actions available at an infoset, \alg\ takes a per-trial time that is equal to the time required to send the reduced strategy to user. \alg\ has a space complexity and initialisation time of only $\mathcal{O}(\nma)$.

\subsection{Related Work}
The notion of differential privacy was introduced in \cite{Dwork2006CalibratingNT}. The notion of local differential privacy was introduced in \cite{Kasiviswanathan2008WhatCW}. The work \cite{Tossou2017AchievingPI} was the first to study differential privacy in the standard adversarial bandit problem (the special case of our problem when we have a single infoset), giving an $\ep$-locally differentially private algorithm with a regret bound of $\tilde{\mathcal{O}}(\sqrt{AT}/\ep)$. The work \cite{Asi2025FasterRF} gave an algorithm for the weaker notion of non-local differential privacy in the standard adversarial bandit problem with a regret bound of $\mathcal{O}(\sqrt{AT/\ep})$. Differentially private reinforcement learning (where transitions are stochastic rather than adversarial) is also related to our work, and was studied in \cite{Balle2016DifferentiallyPP, Garcelon2020LocalDP, Bai2024DifferentiallyPN} among other works.

It should be noted that the works \cite{Farina2021BanditLO, Maiti2025EfficientNA}, which solve the adversarial bandit problem in extensive form games, update based only on the observed loss and not the infosets that are traversed. Due to this fact, we conjecture that the machinery of \cite{Tossou2017AchievingPI} could be applied to these algorithms to make them $\ep$-locally differentially private. However, the regret bound of the modified algorithm of \cite{Farina2021BanditLO} would be $\tilde{\mathcal{O}}(\sqrt{(\varphi^2+H)\nma^3 T}/\ep)$ where $\varphi$ is dilated entropy distance, typically a large quantity. Noting that $\ln(|\rsts|)<\nma$ and is typically substantially less, our bound is dramatically better. Their per-trial time complexity of $\mathcal{O}(\nma)$ is also significantly larger than ours. The modified algorithm of \cite{Maiti2025EfficientNA}, whilst only having a regret bound worse than ours by logarithmic factors, has an extremely high per-trial time complexity of $\mathcal{O}(\nma^3)$ due to the need to solve a convex program.

\alg\ is inspired by the \emph{Laplace mechanism} \cite{Dwork2006CalibratingNT},  \textsc{Exp3-IX} \cite{Neu2015ExploreNM} and \emph{Dilated entropy mirror descent} (\textsc{DMD}) \cite{Hoda2010SmoothingTF, Farina2021BetterRF} but is not a simple combination of these works. Our analysis somewhat is inspired by (although is very different from) that of the \textsc{CanProp} algorithm \cite{Pasteris2023NearestNW} which was designed for efficiently solving contextual bandit problems. The work \cite{Kozuno2021ModelFreeLF} is related to our work as, like us, it constructs hallucinated losses (in a different way from us) which are inputted into \textsc{DMD}.

\subsection{Notation}

Let $\mathbb{N}$ be the set of natural numbers excluding $0$ and, for all $i\in\mathbb{N}$, let $[i]:=\{j\in\mathbb{N}\,|\,j\leq i\}$. Given a predicate $P$ let $\indi{P}:=1$ if $P$ is true and let $\indi{P}:=0$ otherwise. We note that we will sometimes write, for some predicate $P$, $\indi{P}x$ where $x$ is undefined. In all such situations $P$ will be false and hence $\indi{P}x=0$.

\section{The Game}

For this paper we need not fully define an imperfect information extensive-form game as we only need to focus on the learner's infosets and actions and not those of its opponents. In this one-sided view of the game we have a rooted tree whose nodes are either an infoset of the learner, an action of the learner, or a terminal node (a.k.a. \emph{leaf}). The children of each infoset are the set of actions that are available at that infoset and the children of each action are the possible leaves or infosets (of the learner) that can be possibly be encountered after that action is taken.

Formally, we have a rooted tree and a function $\losf:\tns\rightarrow[0,1]$ where $\tns$ is the set of leaves of the tree. Let $\ver$ be the set of nodes of the tree and let $\rot$ be the root of the tree. Given a node $\anod\in\ver$, let $\ch{\anod}$ be the set of its children and, given $\anod\neq\rot$, let $\pnt{\anod}$ be its parent. The set $\ver\setminus\tns$, of internal nodes of the tree, is partitioned into two sets $\lns$ (the set of our \emph{infosets}) and $\nns$ (the set of our \emph{actions}) satisfying the following rules:
\begin{itemize}
    \item $\rot\in\lns$.
    \item For all $\anod\in\lns$ we have $\ch{\anod}\subseteq\nns$.
    \item For all $\ac\in\nns$ we have $\ch{\ac}\subseteq\lns\cup\tns$.
\end{itemize}
Without loss of generality we assume that $|\ch{\anod}|> 1$ for all $\anod\in\lns$.

A \emph{(pure) strategy} is defined as a function $\st:\lns\rightarrow\ver$ with $\st(\anod)\in\ch{\anod}$ for all $\anod\in\lns$. An \emph{environment} is defined as a function $\rf:\nns\rightarrow\ver$ with $\rf(\ac)\in\ch{\ac}$ for all $\ac\in\nns$. Let $\sts$ and $\rfs$ be the sets of all possible strategies and environments respectively. We note that, in this paper, we will, without loss of generality, consider only deterministic environments, as to handle stochastic environments one simply draws an environment from a probability distribution (so our later regret bound will hold for stochastic environments as well).

Given some $\st\in\sts$ and $\rf\in\rfs$, when the pair $(\st, \rf)$ is \emph{played}, a root-to-leaf path through the tree is traversed, and the \emph{loss} $\loss{\st}{\rf}$ is defined as follows. We start at node $\rot$ and:
\begin{itemize}
\item When at a node $\anod\in\lns$, we move next to node $\st(\anod)$.
\item When at a node $\ac\in\nns$\,, we move next to node $\rf(\ac)$.
\item When at a node $\anod\in\tns$ we terminate and define $\loss{\st}{\rf}:=\losf(\anod)$.
\end{itemize}

\subsection{Reduced Strategies}
We now define the notion of a \emph{reduced strategy} which is a strategy that is restricted to the set of nodes in $\lns$ that are reachable via it. A reduced strategy hence contains exactly the information about a strategy that is relevant to playing the game.

Given some $\st\in\sts$ we define $\rable{\st}$ to be the set of nodes that are \emph{reachable} under $\st$ (that is, the set of nodes that can be reached by choosing some $\rf\in\rfs$ and playing $(\st,\rf)$). Formally, $\rable{\st}$ is the minimal subset of $\ver$ in which:
\begin{itemize}
\item $\rot\in\rable{\st}$.
\item For every $\anod\in\lns\cap\rable{\st}$ we have $\st(\anod)\in\rable{\st}$.
\item For every $\ac\in\nns\cap\rable{\st}$ we have $\ch{\ac}\subseteq\rable{\st}$.
\end{itemize}
Given a strategy $\st\in\sts$ we define its corresponding \emph{reduced strategy} $\rst:\rable{\st}\cap\lns\rightarrow\ver$ to be such that $\rst(\anod):=\st(\anod)$ for all $\anod\in\rable{\st}\cap\lns$, and we define $\rable{\rst}:=\rable{\st}$. Let $\rsts$ be the set of all such reduced strategies, noting that this is typically dramatically smaller than $\sts$ as many strategies can share the same reduced strategy. Note that since the domain of a reduced strategy is equal to those nodes in $\lns$ which are reachable via it, we can define, for all $(\st,\rf)\in\rsts\times\rfs$, the value of $\loss{\st}{\rf}$ and how $(\st,\rf)$ is played in exactly the same way as we did for pairs in $\sts\times\rfs$. Given $\st\in\rsts$, define:
\be
\ran{\st}:=\lns\cap\rable{\st}~~~~~,~~~~~\sacs{\st}:=\acs\cap\rable{\st}\,.
\ee

\section{The Problem and Result}\label{parsec}

We have a-priori knowledge of the tree but not necessarily the function $\losf$. We have $T$ trials. We have a server and for each trial $t\in[T]$ we have a user with an unknown environment $\rft{t}\in\rfs$ fixed a-priori. On each trial $t\in[T]$ in turn, the following happens:
\begin{enumerate}
\item The server stochastically chooses some reduced strategy $\stt{t}\in\rsts$ and passes it to the user.
\item The pair $(\stt{t},\rft{t})$ is played. The user observes the nodes in $\acs$ that are traversed (i.e. it knows the actions that it takes) as well as the resulting loss: \be
\lost{t}:=\loss{\stt{t}}{\rft{t}}\,.
\ee
\item The user (based on what it has observed) stochastically selects some data $\data{t}$ from some space $\datas$ prescribed by the algorithm. The user then sends $\data{t}$ to the server.
\end{enumerate}
We have some constant $\ep\in(0,1)$. It is required that for each trial $t$, the data sent by the user is $\ep$-\emph{locally differentially private}. Formally, for any measurable $\datass\subseteq\datas$, any $\sst\in\rsts$ and any $\srf,\srfh\in\rfs$ we must have:
\be
\cprob{\data{t}\in\datass}{\stt{t}=\sst, \rft{t}=\srf} \leq\exp(\ep)\cprob{\data{t}\in\datass}{\stt{t}=\sst, \rft{t}=\srfh}\ ,
\ee
where the probabilities are over the stochasticity of the user's selection of $\data{t}$.

Our performance is measured by the (expected) \emph{regret}, defined as:
\be
\reg := \expt{\sum_{t\in[T]}\lost{t}}-\min_{\st\in\rsts}\sum_{t\in[T]}\loss{\st}{\rft{t}}\ ,
\ee
which is the difference between our expected cumulative loss and that which would have been incurred by always playing the best fixed reduced strategy.

In this paper we give an algorithm \alg\ and prove the following theorem about it.

\begin{theorem}\label{mainth}
\alg\ is $\ep$-locally differentially private and has a regret bound of:
\be
\reg\leq 1+2\sqrt{\left(\frac{6\ln(T)}{\ep}+\frac{9(e-2)}{\ep^2}\right)|\acs|\ln(|\rsts|)T}\,.
\ee
The per-trial time complexity of \alg\ is in:
\be
\mathcal{O}\left(\max_{\st\in\rsts}\sum_{\anod\in\ran{\st}}\ln(|\ch{\anod}|)\right)\,.
\ee
\alg\ requires $\mathcal{O}(|\acs|)$ time to initialise and has a space complexity of $\mathcal{O}(|\acs|)$.
\end{theorem}

Theorem \ref{mainth} is proved in Appendix \ref{anasec}.

\nc{\nsts}{n}
\nc{\nat}{\mathbb{N}}
\nc{\pol}{\pi}
\nc{\polt}[1]{\pi_{#1}}
\nc{\pols}{\mathcal{P}}
\nc{\sample}[1]{\textsc{Sample}_t(#1)}
\nc{\samt}{\textsc{Sample}_t}
\nc{\lne}[1]{z_{#1}}
\nc{\lap}{Q}
\nc{\anum}{x}
\nc{\noise}[2]{\phi_{#1}(#2)}
\nc{\update}[2]{\textsc{Update}_t(#1,#2)}
\nc{\updt}{\textsc{Update}_t}
\nc{\nrm}[2]{\psi_{#1}(#2)}
\nc{\wgt}[2]{\omega_{#1}(#2)}
\nc{\ga}{\gamma}
\nc{\lr}{\eta}
\nc{\ssw}[1]{\psi'_{#1}}
\nc{\msn}{m}
\nc{\ime}{\beta}
\nc{\ptn}[1]{\mathcal{Z}(#1)}
\nc{\ptnp}[2]{\mathcal{Z}'(#1,#2)}
\nc{\sac}{b}

\section{The Algorithm}
Before its formal presentation, we first give an informal overview of the mechanics of \alg.

The server maintains a policy which associates, with each infoset $\anod\in\lns$, a probability distribution over the set of available actions $\ch{\anod}$. This policy varies from trial to trial and we let $\polt{t}$ be its value at the start of trial $t$. The initial policy $\polt{1}$ is constructed via a recursive subroutine. On trial $t$, the server draws the reduced strategy $\stt{t}$ according to $\polt{t}$, in that for each infoset $\anod\in\ran{\stt{t}}$ we have that $\stt{t}(\anod)$ is independently drawn from the probability distribution that $\polt{t}$ associates with $\anod$. Once $(\stt{t},\rft{t})$ is played, the user constructs a function $\data{t}:\sacs{\stt{t}}\rightarrow\mathbb{R}$ by, for every action $\ac\in\sacs{\stt{t}}$\,, independently drawing a number $\noise{t}{\ac}$ from the Laplace distribution (with parameter $2/\ep$) \cite{Dwork2006CalibratingNT} and setting:
\be
\data{t}(\ac):= \indi{\ac=\lne{t}}\lost{t}+\noise{t}{\ac}\ ,
\ee
where $\lne{t}$ is the last action taken by the user (i.e. the last node in $\acs$ that is encountered) when $(\stt{t},\rft{t})$ is played. It is this function $\data{t}$ that is sent back to the server. Inspired by \textsc{Exp3-IX} \cite{Neu2015ExploreNM}, for each action $\ac\in\acs$ we define its \emph{hallucinated loss} as:
\be
\frac{\indi{\ac\in\sacs{\stt{t}}}\data{t}(\ac)}{\ga\ime(\ac)+\cprob{\ac\in\sacs{\stt{t}}}{\polt{t}}}\ ,
\ee
where $\ga\ime(\ac)$ is a specific number that is associated with action $\ac$ at initialisation time (via a recursive subroutine). The server updates $\polt{t}$ to $\polt{t+1}$ via what is essentially the \textsc{DMD} \cite{Farina2021BetterRF} update but using the hallucinated losses instead of the (unknown) counterfactual losses. We note that since the hallucinated losses are zero for all actions not in $\sacs{\stt{t}}$, we can avoid the $\mathcal{O}(|\acs|)$ time complexity of \textsc{DMD}.

We now formally describe our algorithm \alg.

We first construct a function $\nsts:\lns\cup\acs\rightarrow\nat$ recursively (up the tree) as follows:
\begin{itemize}
\item For all $\ac\in\acs$ we have:
\be
\nsts(\ac):=\prod_{\anod\in\ch{\ac}\cap\lns}\nsts(\anod)\,.
\ee
\item For all $\anod\in\lns$ we have:
\be
\nsts(\anod) :=\sum_{\ac\in\ch{\anod}}\nsts(\ac)\,.
\ee
\end{itemize}
We note that, for any node $\anod\in\lns\cup\acs$\,, $\nsts(\anod)$ is the number of reduced sub-strategies (see Section \ref{prsksec}) that start from $\anod$.

We next construct a function $\msn:\lns\cup\acs\rightarrow\mathbb{N}$ recursively (up the tree) as follows:
\begin{itemize}
\item For all $\ac\in\acs$ we have:
\be
\msn(\ac):=1+\sum_{\anod\in\ch{\ac}\cap\lns}\msn(\anod)\,.
\ee
\item For all $\anod\in\lns$ we have:
\be
\msn(\anod):=\sum_{\ac\in\ch{\anod}}\msn(\ac)\,.
\ee
\end{itemize}
We note that for any node $\anod\in\lns\cup\acs$\,, $\msn(\anod)$ is the number of nodes in $\acs$ that are descendants of $\anod$.

We next construct a function $\ime:\lns\cup\acs\rightarrow\mathbb{R}$ recursively (down the tree) as follows:
\begin{itemize}
\item $\ime(\rot):=1$
\item For all $\ac\in\acs$ we have:
\be
\ime(\ac):=\msn(\ac)\ime(\pnt{\ac})\,.
\ee
\item For all $\anod\in\lns\setminus\{\rot\}$ we have:
\be
\ime(\anod):=\frac{\ime(\pnt{\anod})}{\msn(\anod)}\,.
\ee
\end{itemize}

Now define:
\be
\lr:=\left(\left(\frac{6\ln(T)}{\ep}+\frac{9(e-2)}{\ep^2}\right)\frac{\msn(\rot)T}{\ln(\nsts(\rot))}\right)^{-\frac{1}{2}}~~~~~,~~~~~
\ga:=\frac{6\ln(T)\lr}{\ep}\,.
\ee

A \emph{policy} is any function $\pol:\acs\rightarrow[0,1]$ in which, for all $\anod\in\lns$, we have:
\be
\sum_{\ac\in\ch{\anod}}\pol(\ac)=1\,.
\ee
Let $\pols$ be the set of all possible policies.

The server maintains a dynamic (in that it changes from trial to trial) policy. Let $\polt{t}\in\pols$ be this policy at the start of trial $t$. $\polt{1}$ is defined such that for all $\anod\in\lns$ and all $\ac\in\ch{\anod}$ we have:
\be
\polt{1}(\ac):=\frac{\nsts(\ac)}{\nsts(\anod)}\,.
\ee
At the start of trial $t$ the server constructs $\stt{t}$ via the recursive subroutine $\samt$ which takes, as input, a node in $\lns$. Given $\anod\in\lns$, $\sample{\anod}$ is implemented as follows:
\begin{enumerate}
    \item Draw $\stt{t}(\anod)$ from $\ch{\anod}$ such that, for all $\ac\in\ch{\anod}$, the probability that $\stt{t}(\anod)=\ac$ is equal to $\polt{t}(\ac)$.
    \item For all $\anod'\in\ch{\stt{t}(\anod)}\cap\lns$ run the subroutine $\sample{\anod'}$.
\end{enumerate}
Running $\sample{\rot}$ constructs the entire reduced strategy $\stt{t}$.

Once $\stt{t}$ has been constructed by the server it is sent to the user and $(\stt{t},\rft{t})$ is played. Let $\lne{t}$ be the last node in $\acs$ that is encountered when $(\stt{t},\rft{t})$ is played. 

Let $\lap$ be the Laplace distribution with parameter $2/\ep$. i.e. $\lap$ is the probability density function defined by:
\be
\lap(\anum):=\frac{\ep}{4}\exp\left(-\frac{\ep}{2}|\anum|\right)
\ee
for all $\anum\in\mathbb{R}$. The user constructs a function: \be
\data{t}:\sacs{\stt{t}}\rightarrow\mathbb{R}
\ee
by, for all $\ac\in\sacs{\stt{t}}$, drawing a number $\noise{t}{\ac}$ independently from $\lap$ and setting:
\be
\data{t}(\ac):= \indi{\ac=\lne{t}}\lost{t}+\noise{t}{\ac}\,.
\ee
Note then that $\datas$ is the space of all functions that map $\sacs{\st}$ into $\mathbb{R}$ for some $\st\in\rsts$.

The user then sends $\data{t}$ to the server which then constructs $\polt{t+1}$ via the recursive subroutine $\updt$ which takes, as input, a pair in $\ran{\stt{t}}\times\mathbb{R}^+$ and returns a number in $\mathbb{R}^+$. Given some $(\anod,\anum)\in\ran{\stt{t}}\times\mathbb{R}^+$, $\update{\anod}{\anum}$ is implemented as follows:
\begin{enumerate}
\item For each $\anod'\in\ch{\stt{t}(\anod)}\cap\lns$ run $\update{\anod'}{\polt{t}(\stt{t}(\anod))\anum}$ and let $\nrm{t}{\anod'}$ be its return value.
\item Set:
\be
\wgt{t}{\anod}:=\exp\left(\frac{-\lr\data{t}(\stt{t}(\anod))}{\ga\ime(\stt{t}(\anod))+\polt{t}(\stt{t}(\anod))\anum}\right)\prod_{\anod'\in\ch{\stt{t}(\anod)}\cap\lns}\nrm{t}{\anod'}\,.
\ee
\item For all $\ac\in\ch{\anod}$ set:
\be
\polt{t+1}(\ac) := \frac{\indi{\ac=\stt{t}(\anod)}\wgt{t}{\anod}\polt{t}(\ac)+\indi{\ac\neq\stt{t}(\anod)}\polt{t}(\ac)}{1-(1-\wgt{t}{\anod})\polt{t}(\stt{t}(\anod))}\,.
\ee
\item Return $1-(1-\wgt{t}{\anod})\polt{t}(\stt{t}(\anod))$\,.
\end{enumerate}
The server constructs $\polt{t+1}$ by running $\update{\rot}{1}$ and, for all $\anod\in\lns\setminus\ran{\stt{t}}$, retaining:
\be
\polt{t+1}(\ac):=\polt{t}(\ac)~~~\forall \ac\in\ch{\anod}\,.
\ee
Note that, for any $\anod\in\lns$, one can use a segment tree \cite{Sato2025FastEA} in order to, on any trial $t$ with $\anod\in\rable{\stt{t}}$, run $\sample{\anod}$ and $\update{\anod}{\cdot}$ in a time of $\mathcal{O}(\ln(|\ch{\anod}|))$.

\nc{\coms}{\sigma^*}
\nc{\prob}[1]{\mathbb{P}\left[#1\right]}
\nc{\aac}{a}
\nc{\delt}[1]{\Delta_{#1}}
\nc{\comz}[1]{z'_{#1}}
\nc{\exptn}[2]{\mathbb{E}_{\phi_{#1}}\left[#2\right]}
\nc{\cexpt}[2]{\mathbb{E}\left[#1\,|\,#2\right]}
\nc{\des}[1]{\mathcal{N}(#1)}
\nc{\stn}[1]{\mathcal{S}^\dag(#1)}
\nc{\rstn}[1]{\mathcal{S}(#1)}
\nc{\spolt}[1]{\pi'_{#1}}
\nc{\datp}[1]{\data{#1}}
\nc{\mst}{\mu}
\nc{\event}{\mathcal{E}}
\nc{\const}{c}
\nc{\cdata}[1]{d'_{#1}}
\nc{\clf}{y}
\nc{\clfh}{y'}
\nc{\aset}{\mathcal{X}}

\section{Proof Sketch}\label{prsksec}

Here we sketch the proof of Theorem \ref{mainth}, deferring the full proof to Appendix \ref{anasec}. We note that many steps in this sketch are far from immediate - it is intended only to give the reader an overview of how the analysis works. The interested reader can skip straight to Appendix \ref{anasec} as it does not depend on this section.

Note that the proof of the time complexity is virtually immediate, and we hence proceed to the proof of $\ep$-local differential privacy. Consider any trial $t\in[T]$ and take any $\sst\in\rsts$ and $\srf,\srfh\in\rfs$. We condition what follows on the value of $\stt{t}$. Draw a function $\cdata{t}:\sacs{\stt{t}}\rightarrow\mathbb{R}$ such that, for all $\ac\in\sacs{\stt{t}}$, we have that $\cdata{t}(\ac)$ is independently drawn from $\lap$. Let $\clf$ be the last node in $\acs$ that is encountered when $(\stt{t},\srf)$ is played. From the standard analysis of the Laplace mechanism \cite{Dwork2006CalibratingNT} we have, for any measurable $\aset\subseteq\mathbb{R}$ that:
\be
\cprob{\data{t}(\clf)\in\aset}{\rft{t}=\srf}\leq\exp\left(\frac{\ep}{2}\right)\prob{\cdata{t}(\clf)\in\aset}\,.
\ee
We also have, for all $\ac\in\sacs{\stt{t}}\setminus\{\clf\}$ and any measurable $\aset\subseteq\mathbb{R}$, that:
\be
\cprob{\data{t}(\ac)\in\aset}{\rft{t}=\srf}=\prob{\cdata{t}(\ac)\in\aset}\,.
\ee
Hence, for any measurable $\datass\subseteq\datas$ we have:
\begin{equation}\label{psdpeq1}
\cprob{\data{t}\in\aset}{\rft{t}=\srf}\leq\exp\left(\frac{\ep}{2}\right)\prob{\cdata{t}\in\aset}\,.
\end{equation}
In a similar vein we have, for any measurable $\datass\subseteq\datas$, that:
\begin{equation}\label{psdpeq2}
\prob{\cdata{t}\in\aset}\leq\exp\left(\frac{\ep}{2}\right)\cprob{\data{t}\in\aset}{\rft{t}=\srfh}\,.
\end{equation}
Substituting Equation \eqref{psdpeq2} into Equation \eqref{psdpeq1} completes the proof of $\ep$-local differential privacy.

We now turn to the proof of the regret bound. Let:
\be
\coms:=\operatorname{argmin}_{\st\in\rsts}\sum_{t\in[T]}\loss{\st}{\rft{t}}
\ee
and for all $t\in[T]$ define:
\be
\delt{t}=\sum_{\ac\in\sacs{\coms}}\ln(\polt{t}(\ac))\,.
\ee
Without loss of generality we can assume that $T\geq|\acs|$ as the bound clearly holds otherwise. For all $t\in[T]$ and $\ac\in\acs\setminus\sacs{\stt{t}}$ let $\noise{t}{\ac}$ be a random number drawn independently from $\lap$. Let $\event$ be the event that:
\be
|\noise{t}{\ac}|\leq\frac{\ga}{\lr}
\ee
for all $t\in[T]$ and $\ac\in\acs$.
We have:
\be
\expt{\sum_{t\in[T]}\lost{t}}\leq\cexpth{\sum_{t\in[T]}\lost{t}}{\event}+1
\ee
and hence we only need to prove that:
\begin{equation}\label{pseq0}
\cexpth{\sum_{t\in[T]}\lost{t}}{\event}\leq\sum_{t\in[T]}\loss{\coms}{\rft{t}}+2\sqrt{\left(\frac{6\ln(T)}{\ep}+\frac{9(e-2)}{\ep^2}\right)|\acs|\ln(|\rsts|)T}
\end{equation}
so from here on we will implicitly condition everything on the event $\event$. Note that for all $t\in[T]$ and $\ac\in\acs$ we have that $\noise{t}{\ac}$ is still independently drawn with mean $0$.

In order to prove Equation \eqref{pseq0}, consider first some $t\in[T]$. We seek to bound the \emph{expected progress} $\expt{\delt{t+1}-\delt{t}}$. First note that for all $\anod\in\ran{\stt{t}}$ we have, when we let $\ac:=\stt{t}(\anod)$, that:
\begin{equation}\label{pseq1}
\wgt{t}{\anod}=\exp\left(\frac{-\lr\data{t}(\stt{t}(\anod))}{\ga\ime(\stt{t}(\anod))+\cprob{\ac\in\rable{\stt{t}}}{\polt{t}}}\right)\prod_{\anod'\in\ch{\stt{t}(\anod)}\cap\lns}\nrm{t}{\anod'}\,.
\end{equation}
This implies that for all $\anod\in\ran{\coms}$ we have:
\begin{align*}
\ln\left(\frac{\polt{t+1}(\coms(\anod))}{\polt{t}(\coms(\anod))}\right) =& \frac{-\indi{\coms(\anod)\in\rable{\stt{t}}}\lr\data{t}(\coms(\anod))}{\ga\ime(\coms(\anod)) + \cprob{\coms(\anod)\in\rable{\stt{t}}}{\polt{t}}}-\indi{\anod\in\rable{\stt{t}}}\ln(\nrm{t}{\anod})\\
&+\sum_{\anod'\in\ch{\coms(\anod)}\cap\lns}\indi{\anod'\in\rable{\stt{t}}}\ln(\nrm{t}{\anod'})\,.
\end{align*}

By a telescoping sum we then have that:
\begin{align*}
\delt{t}-\delt{t+1} &=-\sum_{\anod\in\ran{\coms}}\ln\left(\frac{\polt{t+1}(\coms(\anod))}{\polt{t}(\coms(\anod))}\right)=\ln(\nrm{t}{\rot}) + \sum_{\ac\in\sacs{\coms}}\frac{\indi{\ac\in\rable{\stt{t}}}\lr\data{t}(\ac)}{\ga\ime(\ac) + \cprob{\ac\in\rable{\stt{t}}}{\polt{t}}}
\end{align*}
which gives us:
\begin{equation}\label{pseq2}
    \expt{\delt{t}-\delt{t+1}}\leq\expt{\ln(\nrm{t}{\rot})}+ \lr\loss{\coms}{\rft{t}}\,.
\end{equation}

In order to bound $\expt{\ln(\nrm{t}{\rot})}$ we make the following definitions.

Given some $\anod\in\ver$, let $\des{\anod}$ be the set of nodes in $\lns$ that are descendants of $\anod$ and let $\stn{\anod}$ be the set of all functions $\st:\des{\anod}\rightarrow\acs$ with $\st(\anod')\in\ch{\anod'}$ for all $\anod'\in\des{\anod}$. We call the elements of such a set $\stn{\anod}$ \emph{sub-strategies}. Note that such a sub-strategy tells us how to play the game from node $\anod$ onwards. As we did for strategies, we now define the concept of a \emph{reduced sub-strategy} which contains exactly the information in a sub-strategy that is relevant to playing the game. Formally, given some $\anod\in\ver$ and $\st\in\stn{\anod}$, define $\rable{\st}$ to be the minimal subset of $\ver$ in which:
\begin{itemize}
\item $\anod\in\rable{\st}$.
\item For every $\anod'\in\des{\anod}\cap\rable{\st}$ we have $\st(\anod')\in\rable{\st}$.
\item For every $\ac\in\nns\cap\rable{\st}$ we have $\ch{\ac}\subseteq\rable{\st}$.
\end{itemize}
 Intuitively, $\rable{\st}$ is the set of nodes that descend from $\anod$ and are reachable when we are resigned to playing $\st$ from node $\anod$ onwards. Given, for some $\anod\in\ver$, a sub-strategy $\st\in\stn{\anod}$ we define its corresponding \emph{reduced sub-strategy} $\rst:\rable{\st}\cap\des{\anod}\rightarrow\acs$ to be such that $\rst(\anod'):=\st(\anod')$ for all $\anod'\in\rable{\st}\cap\des{\anod}$\,, and we define $\rable{\rst}:=\rable{\st}$. Given some $\anod\in\ver$, let $\rstn{\anod}$ be the set of all corresponding reduced sub-strategies of sub-strategies in $\stn{\anod}$, and for all $\st\in\rstn{\anod}$ define:
\be
\ran{\st}:=\lns\cap\rable{\st}~~~~~,~~~~~\sacs{\st}:=\acs\cap\rable{\st}\,.
\ee

Given some $\rf\in\rfs$ we define $\rable{\rf}$ to be the set of nodes that are \emph{reachable} under $\rf$ (that is, the set of nodes that can be reached by choosing some $\st\in\rsts$ and playing $(\st,\rf)$). Formally, $\rable{\rf}$ is the minimal subset of $\ver$ in which:
\begin{itemize}
\item $\rot\in\rable{\rf}$.
\item For every $\anod\in\lns\cap\rable{\rf}$ we have $\ch{\anod}\subseteq\rable{\rf}$.
\item For every $\ac\in\nns\cap\rable{\rf}$ we have $\rf(\ac)\in\rable{\rf}$.
\end{itemize}
Given some $\rf\in\rfs$ we define $\ptn{\rf}$ to be the set of all nodes $\ac\in\acs\cap\rable{\rf}$ in which $\rf(\ac)\in\tns$.

For all $\ac\in\acs$ let:
\be
\vth(\ac):=\frac{\indi{\ac\in\rable{\stt{t}}}\lr\datp{t}(\ac)}{\ga\ime(\ac)+\cprob{\ac\in\rable{\stt{t}}}{\polt{t}}}\,.
\ee
Given $\st\in\rsts$, let $\zti{t}(\st)$ be the last node in $\acs$ that is encountered when $(\st,\rft{t})$ is played.

We now proceed to bound $\expt{\ln(\nrm{t}{\rot})}$. First note that by induction up $\acs$ we have, for all $\ac\in\acs$, that:
\be
\polt{t}(\ac)=\sum_{\st\in\rstn{\ac}}\prod_{\sac\in\sacs{\st}}\polt{t}(\sac)\,.
\ee
Using this equation and Equation \eqref{pseq1} we can prove, by induction up $\ran{\stt{t}}$, that for all $\anod\in\ran{\stt{t}}$ we have:
\be
\nrm{t}{\anod}=\sum_{\st\in\rstn{\anod}}\prod_{\ac\in\sacs{\st}}\polt{t}(\ac)\exp(-\vth(\ac))\,.
\ee
In particular, since $\rstn{\rot}=\rsts$ and:
\be
\prod_{\ac\in\sacs{\st}}\polt{t}(\ac)=\cprob{\stt{t}=\st}{\polt{t}}~~~~~\forall \st\in\rsts
\ee
we obtain:
\begin{equation}\label{pseq3}
\cexpt{\nrm{t}{\rot}}{\polt{t}} = \sum_{\st\in\rsts}\cprob{\stt{t}=\st}{\polt{t}}\cexpth{\exp\left(-\sum_{\ac\in\sacs{\st}}\vth(\ac)\right)}{\polt{t}}\,.
\end{equation}
Now consider any $\st\in\rsts$. We have that:
\be
\sum_{\ac\in\sacs{\st}}\frac{1}{\ime(\ac)}=1
\ee
so since the event $\event$ holds, we have:
\be
\sum_{\ac\in\sacs{\st}}\vth(\ac)\geq -1
\ee
which implies:
\be
\exp\left(-\sum_{\ac\in\sacs{\st}}\vth(\ac)\right)\leq 1 - \sum_{\ac\in\sacs{\st}}\vth(\ac) +(e-2)\sum_{\ac\in\sacs{\st}}\sum_{\ac'\in\sacs{\st}}\vth(\ac)\vth(\ac')\,.
\ee
Taking expectations and noting that:
\be
\frac{\cprob{\zti{t}(\st)\in\rable{\stt{t}}}{\polt{t}}\lr\loss{\st}{\rft{t}}}{\ga\ime(\zti{t}(\st))+\cprob{\zti{t}(\st)\in\rable{\stt{t}}}{\polt{t}}}\geq\lr\loss{\st}{\rft{t}}-\frac{\lr\ga\ime(\zti{t}(\st))}{\cprob{\zti{t}(\st)\in\rable{\stt{t}}}{\polt{t}}}
\ee
then gives us:
\begin{align}
\notag\cexpth{\exp\left(-\sum_{\ac\in\sacs{\st}}\vth(\ac)\right)}{\polt{t}}\leq& 1-\lr\loss{\st}{\rft{t}}+\frac{\lr\ga\ime(\zti{t}(\st))}{\cprob{\zti{t}(\st)\in\rable{\stt{t}}}{\polt{t}}}
\\
\label{pseq4}&+\frac{9(e-2)\lr^2}{\ep^2}\sum_{\ac\in\sacs{\st}}\frac{1}{\cprob{\ac\in\rable{\stt{t}}}{\polt{t}}}\,.
\end{align}

We will now simplify the result of substituting Equation \eqref{pseq4} into Equation \eqref{pseq3}. Note first that:
\begin{equation}\label{pseq5}
\sum_{\st\in\rsts}\cprob{\stt{t}=\st}{\polt{t}}\loss{\st}{\rft{t}}=\cexpt{\lost{t}}{\polt{t}}\,.
\end{equation}
We also have that:
\begin{align}
\notag\sum_{\st\in\rsts}\frac{\ime(\zti{t}(\st))\cprob{\stt{t}=\st}{\polt{t}}}{\cprob{\zti{t}(\st)\in\rable{\stt{t}}}{\polt{t}}}&=\sum_{\ac\in\ptn{\rft{t}}}\frac{\ime(\ac)}{\cprob{\ac\in\rable{\stt{t}}}{\polt{t}}}\sum_{\st\in\rsts}\indi{\ac=\zti{t}(\st)}\cprob{\stt{t}=\st}{\polt{t}}\\
&=\sum_{\ac\in\ptn{\rft{t}}}\ime(\ac)
\label{pseq6}=|\acs|
\end{align}
and that:
\begin{align}
\notag\sum_{\st\in\rsts}\sum_{\ac\in\sacs{\st}}\frac{\cprob{\stt{t}=\st}{\polt{t}}}{\cprob{\ac\in\rable{\stt{t}}}{\polt{t}}}&=\sum_{\ac\in\acs}\frac{1}{\cprob{\ac\in\rable{\stt{t}}}{\polt{t}}}\sum_{\st\in\rsts}\indi{\ac\in\sacs{\st}}\cprob{\stt{t}=\st}{\polt{t}}\\
&=\sum_{\ac\in\acs}1
\label{pseq7}=|\acs|\,.
\end{align}
Substituting Equation \eqref{pseq4} into Equation \eqref{pseq3} and then substituting in equations \eqref{pseq5}, \eqref{pseq6} and \eqref{pseq7} gives us:
\be
\cexpt{\nrm{t}{\rot}}{\polt{t}}\leq1-\lr\cexpt{\lost{t}}{\polt{t}}+\lr\ga|\acs|+\frac{9(e-2)\lr^2}{\ep^2}|\acs|
\ee
which implies:
\be
\expt{\ln(\nrm{t}{\rot})}\leq\lr\ga|\acs|+\frac{9(e-2)\lr^2}{\ep^2}|\acs|-\lr\expt{\lost{t}}\,.
\ee

Now that we have successfully bounded $\expt{\ln(\nrm{t}{\rot})}$ we can substitute it into Equation \eqref{pseq2}, giving us:
\begin{equation}\label{pseq8}
\expt{\delt{t}-\delt{t+1}}\leq\lr\ga|\acs|+\frac{9(e-2)\lr^2}{\ep^2}|\acs|-\lr\expt{\lost{t}}+ \lr\loss{\coms}{\rft{t}}
\end{equation}
as our bound on the expected progress.

Utilising Equation \eqref{pseq8} in a telescoping sum gives us:
\begin{align}
\notag\delt{1}&\leq\expt{\delt{1}-\delt{T+1}}=\sum_{t\in[T]}\expt{\delt{t}-\delt{t+1}}\\
\label{pseq9}&\leq\left(\lr\ga+\frac{9(e-2)\lr^2}{\ep^2}\right)|\acs|T-\lr\sum_{t\in[T]}\expt{\lost{t}}+\lr\sum_{t\in[T]}\loss{\coms}{\rft{t}}
\end{align}
and, by an inductive argument up the tree, we have:
\begin{equation}\label{pseq10}
\delt{1}=-\ln(|\rsts|)\,.
\end{equation}
Substituting Equation \eqref{pseq10} and the values of $\lr$ and $\ga$ into Equation \eqref{pseq9} and rearranging gives us Equation \eqref{pseq0}. This completes the proof sketch.

\bibliographystyle{plain}
\bibliography{bib}

@article{Dwork2006CalibratingNT,
  title={Calibrating Noise to Sensitivity in Private Data Analysis},
  author={Cynthia Dwork and Frank McSherry and Kobbi Nissim and Adam D. Smith},
  journal={J. Priv. Confidentiality},
  year={2006},
  volume={7},
  pages={17-51},
  url={https://api.semanticscholar.org/CorpusID:2468323}
}

@article{Kasiviswanathan2008WhatCW,
  title={What Can We Learn Privately?},
  author={Shiva Prasad Kasiviswanathan and Homin K. Lee and Kobbi Nissim and Sofya Raskhodnikova and Adam D. Smith},
  journal={2008 49th Annual IEEE Symposium on Foundations of Computer Science},
  year={2008},
  pages={531-540},
  url={https://api.semanticscholar.org/CorpusID:1935}
}

@article{Tossou2017AchievingPI,
  title={Achieving Privacy in the Adversarial Multi-Armed Bandit},
  author={Aristide C. Y. Tossou and Christos Dimitrakakis},
  journal={ArXiv},
  year={2017},
  volume={abs/1701.04222},
  url={https://api.semanticscholar.org/CorpusID:10674724}
}

@article{Asi2025FasterRF,
  title={Faster Rates for Private Adversarial Bandits},
  author={Hilal Asi and Vinod Raman and Kunal Talwar},
  journal={ArXiv},
  year={2025},
  volume={abs/2505.21790},
  url={https://api.semanticscholar.org/CorpusID:278959217}
}

@article{Balle2016DifferentiallyPP,
  title={Differentially Private Policy Evaluation},
  author={Borja Balle and Maziar Gomrokchi and Doina Precup},
  journal={ArXiv},
  year={2016},
  volume={abs/1603.02010},
  url={https://api.semanticscholar.org/CorpusID:88252}
}

@article{Garcelon2020LocalDP,
  title={Local Differentially Private Regret Minimization in Reinforcement Learning},
  author={Evrard Garcelon and Vianney Perchet and Ciara Pike-Burke and Matteo Pirotta},
  journal={ArXiv},
  year={2020},
  volume={abs/2010.07778},
  url={https://api.semanticscholar.org/CorpusID:222378459}
}

@inproceedings{Bai2024DifferentiallyPN,
  title={Differentially Private No-regret Exploration in Adversarial Markov Decision Processes},
  author={Shaojie Bai and Lanting Zeng and Chengcheng Zhao and Xiaoming Duan and Mohammad Sadegh Talebi and Peng Cheng and Jiming Chen},
  booktitle={Conference on Uncertainty in Artificial Intelligence},
  year={2024},
  url={https://api.semanticscholar.org/CorpusID:276992321}
}

@article{Farina2021BanditLO,
  title={Bandit Linear Optimization for Sequential Decision Making and Extensive-Form Games},
  author={Gabriele Farina and Robin Schmucker and Tuomas Sandholm},
  journal={ArXiv},
  year={2021},
  volume={abs/2103.04546},
  url={https://api.semanticscholar.org/CorpusID:232146858}
}

@article{Kozuno2021ModelFreeLF,
  title={Model-Free Learning for Two-Player Zero-Sum Partially Observable Markov Games with Perfect Recall},
  author={Tadashi Kozuno and Pierre M'enard and R{\'e}mi Munos and Michal Valko},
  journal={ArXiv},
  year={2021},
  volume={abs/2106.06279},
  url={https://api.semanticscholar.org/CorpusID:235417366}
}

@article{Bai2022NearOptimalLO,
  title={Near-Optimal Learning of Extensive-Form Games with Imperfect Information},
  author={Yunru Bai and Chi Jin and Song Mei and Tiancheng Yu},
  journal={ArXiv},
  year={2022},
  volume={abs/2202.01752},
  url={https://api.semanticscholar.org/CorpusID:246485834}
}

@inproceedings{Fiegel2022AdaptingTG,
  title={Adapting to game trees in zero-sum imperfect information games},
  author={C{\^o}me Fiegel and Pierre M'enard and Tadashi Kozuno and R{\'e}mi Munos and Vianney Perchet and Michal Valko},
  booktitle={International Conference on Machine Learning},
  year={2022},
  url={https://api.semanticscholar.org/CorpusID:255125421}
}

@article{Hoda2010SmoothingTF,
  title={Smoothing Techniques for Computing Nash Equilibria of Sequential Games},
  author={Samid Hoda and Andrew Gilpin and Javier F. Pe{\~n}a and Tuomas Sandholm},
  journal={Math. Oper. Res.},
  year={2010},
  volume={35},
  pages={494-512},
  url={https://api.semanticscholar.org/CorpusID:6932987}
}

@article{Maiti2025EfficientNA,
  title={Efficient Near-Optimal Algorithm for Online Shortest Paths in Directed Acyclic Graphs with Bandit Feedback Against Adaptive Adversaries},
  author={Arnab Maiti and Zhiyuan Fan and Kevin Jamieson and Lillian J. Ratliff and Gabriele Farina},
  journal={ArXiv},
  year={2025},
  volume={abs/2504.00461},
  url={https://api.semanticscholar.org/CorpusID:277467759}
}

@article{Pasteris2023NearestNW,
  title={Nearest Neighbour with Bandit Feedback},
  author={Stephen Pasteris and Chris Hicks and Vasilios Mavroudis},
  journal={ArXiv},
  year={2023},
  volume={abs/2306.13773},
  url={https://api.semanticscholar.org/CorpusID:259251744}
}

@inproceedings{Neu2015ExploreNM,
  title={Explore no more: Improved high-probability regret bounds for non-stochastic bandits},
  author={Gergely Neu},
  booktitle={Neural Information Processing Systems},
  year={2015},
  url={https://api.semanticscholar.org/CorpusID:5846129}
}

@article{Sato2025FastEA,
  title={Fast EXP3 Algorithms},
  author={Ryoma Sato and Shinji Ito},
  journal={ArXiv},
  year={2025},
  volume={abs/2512.11201},
  url={https://api.semanticscholar.org/CorpusID:283883783}
}

@article{Farina2021BetterRF,
  title={Better Regularization for Sequential Decision Spaces: Fast Convergence Rates for Nash, Correlated, and Team Equilibria},
  author={Gabriele Farina and Christian Kroer and Tuomas Sandholm},
  journal={Proceedings of the 22nd ACM Conference on Economics and Computation},
  year={2021},
  url={https://api.semanticscholar.org/CorpusID:235212559}
}


\appendix

\section{Analysis}\label{anasec}

We now prove Theorem \ref{mainth}.

\begin{lemma}\label{diffplem}
\alg\ is $\ep$-locally differentially private.
\end{lemma}

\begin{proof}
Take any $t\in[T]$. All expectations in this proof are conditioned on the value of $\stt{t}$. Draw a function $\cdata{t}:\sacs{\stt{t}}\rightarrow\mathbb{R}$ such that, for all $\ac\in\sacs{\stt{t}}$, we have that $\cdata{t}(\ac)$ is independently drawn from $\lap$.

Take any $\srf,\srfh\in\rfs$. Let $\clf$ be the last node in $\acs$ that is encountered when $(\stt{t},\srf)$ is played. Let $\clfh$ be the last node in $\acs$ that is encountered when $(\stt{t},\srfh)$ is played.

It is a classic result (from \cite{Dwork2006CalibratingNT}) that for any $\anum\in[0,1]$, some $\anum'$ drawn from $\lap$, and any measurable $\aset\subseteq\mathbb{R}$ we have:
\be
\exp\left(-\frac{\ep}{2}\right)\prob{\anum'\in\aset}\leq\prob{\anum+\anum'\in\aset}\leq\exp\left(\frac{\ep}{2}\right)\prob{\anum'\in\aset}\,.
\ee
We hence have, for any measurable $\aset\subseteq\mathbb{R}$, that:
\be
\cprob{\data{t}(\clf)\in\aset}{\rft{t}=\srf}\leq\exp\left(\frac{\ep}{2}\right)\prob{\cdata{t}(\clf)\in\aset}
\ee
and:
\be
\exp\left(-\frac{\ep}{2}\right)\prob{\cdata{t}(\clfh)\in\aset}\leq\cprob{\data{t}(\clfh)\in\aset}{\rft{t}=\srfh}\,.
\ee
We also have, for all $\ac\in\sacs{\stt{t}}\setminus\{\clf\}$ and any measurable $\aset\subseteq\mathbb{R}$, that:
\be
\cprob{\data{t}(\ac)\in\aset}{\rft{t}=\srf}=\prob{\cdata{t}(\ac)\in\aset}
\ee
and have, for all $\ac\in\sacs{\stt{t}}\setminus\{\clfh\}$ and any measurable $\aset\subseteq\mathbb{R}$, that:
\be
\prob{\cdata{t}(\ac)\in\aset}=\cprob{\data{t}(\ac)\in\aset}{\rft{t}=\srfh}\,.
\ee
Hence, for any measurable $\datass\subseteq\datas$ we have:
\be
\cprob{\data{t}\in\aset}{\rft{t}=\srf}\leq\exp\left(\frac{\ep}{2}\right)\prob{\cdata{t}\in\aset}
\ee
and:
\be
\exp\left(-\frac{\ep}{2}\right)\prob{\cdata{t}\in\aset}\leq\cprob{\data{t}\in\aset}{\rft{t}=\srfh}
\ee
so that:
\be
\cprob{\data{t}\in\aset}{\rft{t}=\srf}\leq\exp(\ep)\cprob{\data{t}\in\aset}{\rft{t}=\srfh}
\ee
as required.
\end{proof}

\begin{lemma}\label{comcomlem}
The per trial time complexity of \alg\ is in:
\be
\mathcal{O}\left(\max_{\st\in\rsts}\sum_{\anod\in\ran{\st}}\ln(|\ch{\anod}|)\right)\,.
\ee
\end{lemma}

\begin{proof}
Note first that the time required to send $\stt{t}$ to the user on trial $t$, and for the user to construct $\data{t}$ and send it back, is in:
\be
\mathcal{O}\left(\max_{\st\in\rsts}|\ran{\st}|\right)\,.
\ee
Further note that on each trial $t\in[T]$ we run $\sample{\anod}$ and $\update{\anod}{\cdot}$ only for nodes in $\ran{\stt{t}}$. So since we can implement the algorithm such that, for all $t\in[T]$ and $\anod\in\ran{\stt{t}}$, $\sample{\anod}$ and $\update{\anod}{\cdot}$ run in a time of $\mathcal{O}(\ln(|\ch{\anod}|))$, we have the result.
\end{proof}

We now turn to the proof of the regret bound. Without loss of generality we can assume that $T\geq|\acs|$ as the bound clearly holds otherwise. For all $t\in[T]$ and $\ac\in\acs\setminus\sacs{\stt{t}}$ let $\noise{t}{\ac}$ be a random number drawn independently from $\lap$. Let $\event$ be the event that:
\be
|\noise{t}{\ac}|\leq\frac{\ga}{\lr}
\ee
for all $t\in[T]$ and $\ac\in\acs$.

\begin{lemma}\label{eventlem}
We have:
\be
\expt{\sum_{t\in[T]}\lost{t}}\leq\cexpth{\sum_{t\in[T]}\lost{t}}{\event}+1\,.
\ee
\end{lemma}

\begin{proof}
Given $t\in[T]$ and $\ac\in\acs$ we have, since $\lap(\anum)=\lap(-\anum)$ for all $\anum\in\mathbb{R}$, that:
\begin{align*}
\prob{|\noise{t}{\ac}|>\frac{\ga}{\lr}}&=2\int_{\anum>\ga/\lr}\lap(\anum)\\
&=\frac{\ep}{2}\int_{\anum>\ga/\lr}\exp\left(-\frac{\ep}{2}\anum\right)\\
&=\exp\left(-\frac{\ep\ga}{2\lr}\right)\\
&=\exp(-3\ln(T))\\
&=\frac{1}{T^{3}}\,.
\end{align*}
By the union bound we then have:
\begin{align*}
1-\prob{\event}&=\prob{\exists (t,\ac)\in[T]\times\acs\,:\, |\noise{t}{\ac}|>\frac{\ga}{\lr}}\\
&\leq\sum_{t\in[T]}\sum_{\ac\in\acs}\prob{|\noise{t}{\ac}|>\frac{\ga}{\lr}}\\
&=\frac{T|\acs|}{T^{3}}\\
&\leq\frac{1}{T}\,.
\end{align*}
Since we always have:
\be
\sum_{t\in[T]}\lost{t}\leq T
\ee
we then have:
\begin{align*}
\expt{\sum_{t\in[T]}\lost{t}}&\leq\prob{\event}\cexpth{\sum_{t\in[T]}\lost{t}}{\event}+(1-\prob{\event})T\\
&\leq\cexpth{\sum_{t\in[T]}\lost{t}}{\event}+1
\end{align*}
as required.
\end{proof}
Let:
\be
\coms:=\operatorname{argmin}_{\st\in\rsts}\sum_{t\in[T]}\loss{\st}{\rft{t}}\,.
\ee
By Lemma \ref{eventlem} we now only need to prove that:
\be
\cexpth{\sum_{t\in[T]}\lost{t}}{\event}\leq\sum_{t\in[T]}\loss{\coms}{\rft{t}}+2\sqrt{\left(\frac{6\ln(T)}{\ep}+\frac{9(e-2)}{\ep^2}\right)|\acs|\ln(|\rsts|)T}
\ee
so from here on we will implicitly condition everything on the event $\event$. Note that for all $t\in[T]$ and $\ac\in\acs$ we have that $\noise{t}{\ac}$ is still independently drawn with mean $0$.

For any $t\in [T]$ we use the notation $\mathbb{E}_{\phi_{t}}$ to denote the expectation over only the draw of the function $\phi_{t}$ (conditioned on the event $\event$).

For all $t\in[T]$ define:
\be
\delt{t}=\sum_{\ac\in\sacs{\coms}}\ln(\polt{t}(\ac))\,.
\ee

\begin{lemma}\label{omegatlem}
For all $t\in[T]$ and $\anod\in\ran{\stt{t}}$ we have:
\be
\wgt{t}{\anod}=\exp\left(\frac{-\lr\data{t}(\stt{t}(\anod))}{\ga\ime(\stt{t}(\anod))+\cprob{\aac\in\rable{\stt{t}}}{\polt{t}}}\right)\prod_{\anod'\in\ch{\stt{t}(\anod)}\cap\lns}\nrm{t}{\anod'}
\ee
where $\aac:=\stt{t}(\anod)$.
\end{lemma}

\begin{proof}
By induction on the depth of $\anod$ we have, when $\update{\anod}{\anum}$ is run, that:
\be
\anum = \cprob{\anod\in\rable{\stt{t}}}{\polt{t}}
\ee
from which the result follows.
\end{proof}

\begin{lemma}\label{lnpiopilem}
For all $t\in[T]$ and $\anod\in\ran{\coms}$ we have:
\begin{align*}
\ln\left(\frac{\polt{t+1}(\coms(\anod))}{\polt{t}(\coms(\anod))}\right) =& \frac{-\indi{\coms(\anod)\in\rable{\stt{t}}}\lr\data{t}(\coms(\anod))}{\ga\ime(\coms(\anod)) + \cprob{\coms(\anod)\in\rable{\stt{t}}}{\polt{t}}}-\indi{\anod\in\rable{\stt{t}}}\ln(\nrm{t}{\anod})\\
&+\sum_{\anod'\in\ch{\coms(\anod)}\cap\lns}\indi{\anod'\in\rable{\stt{t}}}\ln(\nrm{t}{\anod'})\,.
\end{align*}
\end{lemma}

\begin{proof}
We have two cases - that either $\anod\in\rable{\stt{t}}$ or $\anod\notin\rable{\stt{t}}$.

First consider the case that $\anod\notin\rable{\stt{t}}$. In this case we have, directly from the algorithm, that $\polt{t+1}(\coms(\anod))=\polt{t}(\coms(\anod))$ so since we have $\coms(\anod)\notin\rable{\stt{t}}$, $\anod\notin\rable{\stt{t}}$ and $\anod'\notin\rable{\stt{t}}$ for all $\anod'\in\ch{\coms(\anod)}\cap\lns$, we have the result.

Now let us consider the case that $\anod\in\rable{\stt{t}}$. First note that directly from the algorithm we have:
\be
\nrm{t}{\anod} = 1-(1-\wgt{t}{\anod})\polt{t}(\stt{t}(\anod))
\ee
so that:
\be
\ln\left(\frac{\polt{t+1}(\coms(\anod))}{\polt{t}(\coms(\anod))}\right)=\indi{\coms(\anod)=\stt{t}(\anod)}\ln(\wgt{t}{\anod}) - \ln(\nrm{t}{\anod})\,.
\ee
The result then follows from lemma \ref{omegatlem} and noting that:
\be
\indi{\anod\in\rable{\stt{t}} \wedge \coms(\anod)=\stt{t}(\anod)} = \indi{\coms(\anod)\in\rable{\stt{t}}}
\ee
and:
\be
\indi{\anod\in\rable{\stt{t}} \wedge \coms(\anod)=\stt{t}(\anod)} = \indi{\anod'\in\rable{\stt{t}}}~~~~~ \forall \anod'\in\ch{\coms(\anod)}\cap\lns\,.
\ee
\end{proof}

\begin{lemma}\label{delmdellem1}
For all $t\in[T]$ we have:
\be
\expt{\delt{t}-\delt{t+1}}\leq \lr\loss{\coms}{\rft{t}}+\expt{\ln(\nrm{t}{\rot})}\,.
\ee
\end{lemma}

\begin{proof}
Note first that:
\be
\bigcup_{\anod\in\ran{\coms}}\ch{\coms(\anod)}\cap\lns = \ran{\coms}\setminus\{\rot\}
\ee
and the sets in:
\be
\{\ch{\coms(\anod)}\cap\lns\,|\,\anod\in\ran{\coms}\}
\ee
are pairwise disjoint. So, since $\rot\in\rable{\stt{t}}$, we have:
\be
\sum_{\anod\in\ran{\coms}}\indi{\anod\in\rable{\stt{t}}}\ln(\nrm{t}{\anod})=\ln(\nrm{t}{\rot})+\sum_{\anod\in\ran{\coms}}\sum_{\anod'\in\ch{\coms(\anod)}\cap\lns}\indi{\anod'\in\rable{\stt{t}}}\ln(\nrm{t}{\anod'})
\ee
and hence, by Lemma \ref{lnpiopilem}, we have:
\begin{align}
\notag\delt{t}-\delt{t+1} &=-\sum_{\ac\in\sacs{\coms}}\ln\left(\frac{\polt{t+1}(\ac)}{\polt{t}(\ac)}\right)\\ 
\notag&=-\sum_{\anod\in\ran{\coms}}\ln\left(\frac{\polt{t+1}(\coms(\anod))}{\polt{t}(\coms(\anod))}\right)\\
\notag&=\ln(\nrm{t}{\rot}) + \sum_{\anod\in\ran{\coms}}\frac{\indi{\coms(\anod)\in\rable{\stt{t}}}\lr\data{t}(\coms(\anod))}{\ga\ime(\coms(\anod)) + \cprob{\coms(\anod)\in\rable{\stt{t}}}{\polt{t}}}\\
\label{delmdellem1eq1}&=\ln(\nrm{t}{\rot}) + \sum_{\ac\in\sacs{\coms}}\frac{\indi{\ac\in\rable{\stt{t}}}\lr\data{t}(\ac)}{\ga\ime(\ac) + \cprob{\ac\in\rable{\stt{t}}}{\polt{t}}}\,.
\end{align}
Let $\comz{t}$ be the last node in $\acs$ that is encountered when $(\coms, \rft{t})$ is played. For all $\ac\in\sacs{\coms}\setminus\{\comz{t}\}$ with $\ac\in\rable{\stt{t}}$ we must have that $\ac\neq\lne{t}$ so that $\exptn{t}{\data{t}(\ac)}=0$. Also, if $\comz{t}\in\rable{\stt{t}}$ then we must have that $\comz{t}=\lne{t}$ so that $\exptn{t}{\data{t}(\comz{t})}=\lost{t}$ which in turn (since $\comz{t}=\lne{t}$) is equal to $\loss{\coms}{\rft{t}}$. Hence, Equation \eqref{delmdellem1eq1} and linearity of expectation gives:
\be
\exptn{t}{\delt{t}-\delt{t+1}}=\exptn{t}{\ln(\nrm{t}{\rot})}+\frac{\indi{\comz{t}\in\rable{\stt{t}}}\lr\loss{\coms}{\rft{t}}}{\ga\ime(\comz{t}) + \cprob{\comz{t}\in\rable{\stt{t}}}{\polt{t}}}
\ee
which gives us:
\be
\cexpt{\delt{t}-\delt{t+1}}{\polt{t}}=\cexpt{\ln(\nrm{t}{\rot})}{\polt{t}}+\frac{\cprob{\comz{t}\in\rable{\stt{t}}}{\polt{t}}\lr\loss{\coms}{\rft{t}}}{\ga\ime(\comz{t}) + \cprob{\comz{t}\in\rable{\stt{t}}}{\polt{t}}}\,.
\ee
Since $\lr\loss{\coms}{\rft{t}}\geq0$ and $\ga\ime(\comz{t})>0$ we then have:
\be
\cexpt{\delt{t}-\delt{t+1}}{\polt{t}}\leq\cexpt{\ln(\nrm{t}{\rot})}{\polt{t}}+\lr\loss{\coms}{\rft{t}}
\ee
which implies the result.
\end{proof}

Given some $\anod\in\ver$, let $\des{\anod}$ be the set of nodes in $\lns$ that are descendants of $\anod$ and let $\stn{\anod}$ be the set of all functions $\st:\des{\anod}\rightarrow\acs$ with $\st(\anod')\in\ch{\anod'}$ for all $\anod'\in\des{\anod}$. We call the elements of such a set $\stn{\anod}$ \emph{sub-strategies}. As we did for strategies, we make the following definitions. Given some $\anod\in\ver$ and $\st\in\stn{\anod}$, define $\rable{\st}$ to be the minimal subset of $\ver$ in which:
\begin{itemize}
\item $\anod\in\rable{\st}$.
\item For every $\anod'\in\des{\anod}\cap\rable{\st}$ we have $\st(\anod')\in\rable{\st}$.
\item For every $\ac\in\nns\cap\rable{\st}$ we have $\ch{\ac}\subseteq\rable{\st}$.
\end{itemize}
Given, for some $\anod\in\ver$, a sub-strategy $\st\in\stn{\anod}$ we define its corresponding \emph{reduced sub-strategy} $\rst:\rable{\st}\cap\des{\anod}\rightarrow\acs$ to be such that $\rst(\anod'):=\st(\anod')$ for all $\anod'\in\rable{\st}\cap\des{\anod}$\,, and we define $\rable{\rst}:=\rable{\st}$. Given some $\anod\in\ver$, let $\rstn{\anod}$ be the set of all corresponding reduced sub-strategies of sub-strategies in $\stn{\anod}$, and for all $\st\in\rstn{\anod}$ define:
\be
\ran{\st}:=\lns\cap\rable{\st}~~~~~,~~~~~\sacs{\st}:=\acs\cap\rable{\st}\,.
\ee
Note that $\rstn{\rot}=\rsts$.

\begin{lemma}\label{mrnrlem}
We have:
\be
\msn(\rot)=|\acs|~~~~~,~~~~~
\nsts(\rot)=|\rsts|\,.
\ee
\end{lemma}

\begin{proof}
By induction up the tree we have, for all $\anod\in\lns\cup\acs$\,, that $\msn(\anod)$ is the number of nodes in $\acs$ that are descendants of $\anod$, and that:
\be
\nsts(\anod)=|\rstn{\anod}|\,.
\ee
The result follows.
\end{proof}

\begin{lemma}\label{spiplem}
For all $\ac\in\acs$ we have:
\be
\sum_{\st\in\rstn{\ac}}\prod_{\ac'\in\sacs{\st}}\polt{t}(\ac') = \polt{t}(\ac)\,.
\ee
\end{lemma}

\begin{proof}
We prove by induction up the tree. Hence, we can assume that for all $\anod\in\ch{\ac}\cap\lns$ and for all $\sac\in\ch{\anod}$ we have:
\be
\sum_{\st\in\rstn{\sac}}\prod_{\ac'\in\sacs{\st}}\polt{t}(\ac') = \polt{t}(\sac)\,.
\ee
Note that for all $\anod\in\ch{\ac}\cap\lns$ we have:
\begin{align*}
\sum_{\st\in\rstn{\anod}}\prod_{\ac'\in\sacs{\st}}\polt{t}(\ac')&=\sum_{\sac\in\ch{\anod}}\sum_{\st\in\rstn{\sac}}\prod_{\ac'\in\sacs{\st}}\polt{t}(\ac')\\
&=\sum_{\sac\in\ch{\anod}}\polt{t}(\sac)\\
&=1
\end{align*}
and hence:
\begin{align*}
\sum_{\st\in\rstn{\ac}}\prod_{\ac'\in\sacs{\st}}\polt{t}(\ac')&=\polt{t}(\ac)\prod_{\anod\in\ch{\ac}\cap\lns}\sum_{\st\in\rstn{\anod}}\prod_{\ac'\in\sacs{\st}}\polt{t}(\ac')\\
&=\polt{t}(\ac)\prod_{\anod\in\ch{\ac}\cap\lns}1\\
&=\polt{t}(\ac)
\end{align*}
as required.
\end{proof}

For all $t\in[T]$ and $\ac\in\acs$ let:
\be
\vth(\ac):=\frac{\indi{\ac\in\rable{\stt{t}}}\lr\datp{t}(\ac)}{\ga\ime(\ac)+\cprob{\ac\in\rable{\stt{t}}}{\polt{t}}}\,.
\ee

\begin{lemma}\label{spethlem}
For all $t\in[T]$ we have:
\be
\nrm{t}{\rot} = \sum_{\st\in\rsts}\cprob{\stt{t}=\st}{\polt{t}}\exp\left(-\sum_{\ac\in\sacs{\st}}\vth(\ac)\right)\,.
\ee
\end{lemma}

\begin{proof}
For all $\anod\in\ver$ and $\st\in\rstn{\anod}$ define:
\be
\ssw{t}(\st):=\prod_{\ac\in\sacs{\st}}\polt{t}(\ac)\exp(-\vth(\ac))\,.
\ee
We first prove, by induction up $\ran{\stt{t}}$, that for all $\anod\in\ran{\stt{t}}$ we have:
\be
\nrm{t}{\anod} = \sum_{\st\in\rstn{\anod}}\ssw{t}(\st)\,.
\ee
To prove this by induction we can assume that for all $\anod'\in\ch{\stt{t}(\anod)}\cap\lns$ we have:
\be
\nrm{t}{\anod'} = \sum_{\st\in\rstn{\anod'}}\ssw{t}(\st)\,.
\ee
Let $\sac:=\stt{t}(\anod)$. From Lemma \ref{omegatlem} we have, since $\sac\in\rable{\stt{t}}$, that:
\begin{align*}
\polt{t}(\sac)\wgt{t}{\anod}&=\polt{t}(\sac)\exp\left(\frac{-\lr\data{t}(\sac)}{\ga\ime(\sac)+\cprob{\sac\in\rable{\stt{t}}}{\polt{t}}}\right)\prod_{\anod'\in\ch{\sac}\cap\lns}\nrm{t}{\anod'}\\
&=\polt{t}(\sac)\exp(-\vth(\sac))\prod_{\anod'\in\ch{\sac}\cap\lns}\sum_{\st\in\rstn{\anod'}}\ssw{t}(\st)\\
&=\sum_{\st\in\rstn{\sac}}\ssw{t}(\st)\,.
\end{align*}
Now take any $\sac'\in\ch{v}\setminus\{\sac\}$. Note that for any $\st\in\rstn{\sac'}$ and any $\ac\in\sacs{\st}$ we have that $\ac\notin\rable{\stt{t}}$ and hence, by Lemma \ref{spiplem}, we have:
\begin{align*}
\polt{t}(\sac')
&=\sum_{\st\in\rstn{\sac'}}\prod_{\ac\in\sacs{\st}}\polt{t}(\ac)\\
&=\sum_{\st\in\rstn{\sac'}}\prod_{\ac\in\sacs{\st}}\polt{t}(\ac)\exp(-\vth(\ac))\\
&=\sum_{\st\in\rstn{\sac'}}\ssw{t}(\st)\,.
\end{align*}
We have hence shown that:
\begin{align*}
\nrm{t}{\anod}&=1-(1-\wgt{t}{\anod})\polt{t}(\sac)\\
&=\sum_{\ac\in\ch{\anod}}(\indi{\ac=\sac}\wgt{t}{\anod}\polt{t}(\sac)+\indi{\ac\neq\sac}\polt{t}(\ac))\\
&=\sum_{\ac\in\ch{\anod}}\sum_{\st\in\rstn{\ac}}\ssw{t}(\st)\\
&=\sum_{\st\in\rstn{\anod}}\ssw{t}(\st)
\end{align*}
as required.

We have now shown that the inductive hypothesis holds for all $\anod\in\ran{\stt{t}}$. In particular, we have:
\begin{align*}
\nrm{t}{\rot} &= \sum_{\st\in\rstn{\rot}}\ssw{t}(\st)\\
&=\sum_{\st\in\rsts}\ssw{t}(\st)\\
&=\sum_{\st\in\rsts}\left(\prod_{\ac\in\sacs{\st}}\polt{t}(\ac)\right)\exp\left(-\sum_{\ac\in\sacs{\st}}\vth(\ac)\right)\\
&=\sum_{\st\in\rsts}\cprob{\stt{t}=\st}{\polt{t}}\exp\left(-\sum_{\ac\in\sacs{\st}}\vth(\ac)\right)
\end{align*}
as required.
\end{proof}

\nc{\acd}[2]{\mathcal{A}'(#1,#2)}

\begin{lemma}\label{sooimelem}
For all $\st\in\rsts$ we have:
\be
\sum_{\ac\in\sacs{\st}}\frac{1}{\ime(\ac)}=1\,.
\ee
\end{lemma}

\begin{proof}
Given some $\anod\in\ran{\st}$ let $\acd{\st}{\anod}$ be the set of all nodes in $\sacs{\st}$ that are descendants of $\anod$. We take the inductive hypothesis that for any $\anod\in\ran{\st}$ we have:
\be
\sum_{\ac\in\acd{\st}{\anod}}\frac{1}{\ime(\ac)}=\frac{1}{\ime(\anod)}\,.
\ee
To prove that the inductive hypothesis holds, we need only show that it holds whenever, for all $\anod'\in\ch{\st(\anod)}\cap\lns$, we have:
\be
\sum_{\ac\in\acd{\st}{\anod'}}\frac{1}{\ime(\ac)}=\frac{1}{\ime(\anod')}\,.
\ee
So assume this. We then have:
\begin{align*}
\sum_{\ac\in\acd{\st}{\anod}}\frac{1}{\ime(\ac)}&=\frac{1}{\ime(\st(\anod))}+\sum_{\anod'\in\ch{\st(\anod)}\cap\lns}\sum_{\ac\in\acd{\st}{\anod'}}\frac{1}{\ime(\ac)}\\
&=\frac{1}{\ime(\st(\anod))}+\sum_{\anod'\in\ch{\st(\anod)}\cap\lns}\frac{1}{\ime(\anod')}\\
&=\frac{1}{\ime(\st(\anod))}+\sum_{\anod'\in\ch{\st(\anod)}\cap\lns}\frac{\msn(\anod')}{\ime(\st(\anod))}\\
&=\frac{1}{\ime(\st(\anod))}\left(1+\sum_{\anod'\in\ch{\st(\anod)}\cap\lns}\msn(\anod')\right)\\
&=\frac{\msn(\st(\anod))}{\ime(\st(\anod))}\\
&=\frac{1}{\ime(\anod)}
\end{align*}
as required. We have now shown that the inductive hypothesis holds for all $\anod\in\ran{\st}$ and hence that it holds for $\anod=\rot$ which gives us the result.
\end{proof}

Given some $t\in[T]$ and $\st\in\rsts$, let $\zti{t}(\st)$ be the last node in $\acs$ that is encountered when $(\st,\rft{t})$ is played.

\begin{lemma}\label{exlnnrmlem}
For all $t\in[T]$ and $\st\in\rsts$ we have:
\begin{align*}
\cexpth{\exp\left(-\sum_{\ac\in\sacs{\st}}\vth(\ac)\right)}{\polt{t}}\leq& 1-\lr\loss{\st}{\rft{t}}+\frac{\lr\ga\ime(\zti{t}(\st))}{\cprob{\zti{t}(\st)\in\rable{\stt{t}}}{\polt{t}}}
\\
&+\frac{9(e-2)\lr^2}{\ep^2}\sum_{\ac\in\sacs{\st}}\frac{1}{\cprob{\ac\in\rable{\stt{t}}}{\polt{t}}}\,.
\end{align*}
\end{lemma}

\begin{proof}

Note that since everything is conditioned on the event $\event$ we have that:
\be
\datp{t}(\ac)\geq -\ga/\lr
\ee
for all $\ac\in\sacs{\stt{t}}$. Hence, by Lemma \ref{sooimelem}, we have that:
\begin{align*}
\sum_{\ac\in\sacs{\st}}\vth(\ac)&\geq\sum_{\ac\in\sacs{\st}}\frac{-\ga}{\ga\ime(\ac)+\cprob{\ac\in\rable{\stt{t}}}{\polt{t}}}\\
&\geq\sum_{\ac\in\sacs{\st}}\frac{-1}{\ime(\ac)}\\
&= -1
\end{align*}
so since:
\be
\exp(-\anum)\leq 1 - \anum + (e-2)\anum^2~~~~~~~\forall x\geq-1
\ee
we have:
\begin{equation}\label{exlnnrmlemeq1}
\exp\left(-\sum_{\ac\in\sacs{\st}}\vth(\ac)\right)\leq 1 - \sum_{\ac\in\sacs{\st}}\vth(\ac) +(e-2)\sum_{\ac\in\sacs{\st}}\sum_{\ac'\in\sacs{\st}}\vth(\ac)\vth(\ac')\,.
\end{equation}
For any distinct $\ac,\ac'\in\sacs{\st}$ with $\ac,\ac'\in\rable{\stt{t}}$ we have that either $\ac$ or $\ac'$ is not equal to $\lne{t}$. Without loss of generality, assume then that $\ac\neq\lne{t}$. Then we must have that, under the draw of $\nosf{t}$, $\datp{t}(\ac)$ has mean zero and is independent of $\datp{t}(\ac')$ which means that:
\be
\exptn{t}{\vth(\ac)\vth(\ac')}=0\,.
\ee
It is a well known result that $\lap$ has variance $8/\ep^2$. Hence, for any $\ac\in\sacs{\st}$ with $\ac\in\rable{\stt{t}}$ we have that:
\begin{align*}
\exptn{t}{\vth(\ac)^2}&\leq\frac{\lr^2(1+8/\ep^2)}{(\ga\ime(\ac)+\cprob{\ac\in\rable{\stt{t}}}{\polt{t}})^2}\\
&\leq\frac{9\lr^2}{\ep^2\cprob{\ac\in\rable{\stt{t}}}{\polt{t}}^2}\,.
\end{align*}
Note now that if $\zti{t}(\st)\in\rable{\stt{t}}$ then $\lne{t}=\zti{t}(\st)$ so: 
\be
\lost{t}=\loss{\st}{\rft{t}}
\ee
and that for all $\ac\in\sacs{\st}\setminus\{\zti{t}(\st)\}$ we have $\ac\neq\lne{t}$.
So for all $\ac\in\sacs{\st}$ with $\ac\in\rable{\stt{t}}$, we have:
\be
\exptn{t}{\vth(\ac)}=\frac{\indi{\ac=\zti{t}(\st)}\lr\loss{\st}{\rft{t}}}{\ga\ime(\zti{t}(\st))+\cprob{\zti{t}(\st)\in\rable{\stt{t}}}{\polt{t}}}\,.
\ee
Substituting all these results into equation \eqref{exlnnrmlemeq1} (after taking expectations over $\nosf{t}$) gives us:
\begin{align*}
\exptn{t}{\exp\left(-\sum_{\ac\in\sacs{\st}}\vth(\ac)\right)}\leq& 1 - \frac{\indi{\zti{t}(\st)\in\rable{\stt{t}}}\lr\loss{\st}{\rft{t}}}{\ga\ime(\zti{t}(\st))+\cprob{\zti{t}(\st)\in\rable{\stt{t}}}{\polt{t}}}\\
&+\frac{9(e-2)\lr^2}{\ep^2}\sum_{\ac\in\sacs{\st}}\frac{\indi{\ac\in\rable{\stt{t}}}}{\cprob{\ac\in\rable{\stt{t}}}{\polt{t}}^2}
\end{align*}
which further gives us:
\begin{align*}
\cexpth{\exp\left(-\sum_{\ac\in\sacs{\st}}\vth(\ac)\right)}{\polt{t}}\leq& 1-\frac{\cprob{\zti{t}(\st)\in\rable{\stt{t}}}{\polt{t}}\lr\loss{\st}{\rft{t}}}{\ga\ime(\zti{t}(\st))+\cprob{\zti{t}(\st)\in\rable{\stt{t}}}{\polt{t}}}\\
&+\frac{9(e-2)\lr^2}{\ep^2}\sum_{\ac\in\sacs{\st}}\frac{1}{\cprob{\ac\in\rable{\stt{t}}}{\polt{t}}}\,.
\end{align*}
Noting that:
\begin{align*}
\lr\loss{\st}{\rft{t}}-\frac{\cprob{\zti{t}(\st)\in\rable{\stt{t}}}{\polt{t}}\lr\loss{\st}{\rft{t}}}{\ga\ime(\zti{t}(\st))+\cprob{\zti{t}(\st)\in\rable{\stt{t}}}{\polt{t}}}&=\frac{\lr\ga\ime(\zti{t}(\st))\loss{\st}{\rft{t}}}{\ga\ime(\zti{t}(\st))+\cprob{\zti{t}(\st)\in\rable{\stt{t}}}{\polt{t}}}\\
&\leq\frac{\lr\ga\ime(\zti{t}(\st))}{\cprob{\zti{t}(\st)\in\rable{\stt{t}}}{\polt{t}}}
\end{align*}
the result follows.
\end{proof}

Given some $\rf\in\rfs$ we define $\rable{\rf}$ to be the set of nodes that are \emph{reachable} under $\rf$ (that is, the set of nodes that can be reached by choosing some $\st\in\rsts$ and playing $(\st,\rf)$). Formally, $\rable{\rf}$ is the minimal subset of $\ver$ in which:
\begin{itemize}
\item $\rot\in\rable{\rf}$.
\item For every $\anod\in\lns\cap\rable{\rf}$ we have $\ch{\anod}\subseteq\rable{\rf}$.
\item For every $\ac\in\nns\cap\rable{\rf}$ we have $\rf(\ac)\in\rable{\rf}$.
\end{itemize}
Given some $\rf\in\rfs$ we define $\ptn{\rf}$ to be the set of all nodes $\ac\in\acs\cap\rable{\rf}$ in which $\rf(\ac)\in\tns$.

\begin{lemma}\label{imestalem}
For all $\rf\in\rfs$ we have:
\be
\sum_{\ac\in\ptn{\rf}}\ime(\ac)=|\acs|\,.
\ee
\end{lemma}

\begin{proof}
Given some $\ac\in\rable{\rf}\cap\acs$ let $\ptnp{\rf}{\ac}$ be the set of all nodes in $\ptn{\rf}$ that are descendants of $\ac$. We take the inductive hypothesis that for any $\ac\in\rable{\rf}\cap\acs$ we have:
\be
\sum_{\sac\in\ptnp{\rf}{\ac}}\ime(\sac)=\ime(\ac)\,.
\ee
To prove that the inductive hypothesis holds, we need only show that it holds whenever, for all $\ac'\in\ch{\rf(\ac)}$\,, we have:
\be
\sum_{\sac\in\ptnp{\rf}{\ac'}}\ime(\sac)=\ime(\ac')\,.
\ee
So assume this. We have two cases - that either $\rf(\ac)\in\tns$ or $\rf(\ac)\notin\tns$. In the first case we have $\rf(\ac)\in\tns$ which implies that $\ptnp{\rf}{\ac}=\{\ac\}$, proving the inductive hypothesis. Now let us consider the case that $\rf(\ac)\notin\tns$. In this case we have:
\begin{align*}
\sum_{\sac\in\ptnp{\rf}{\ac}}\ime(\sac)&=\sum_{\ac'\in\ch{\rf(\ac)}}\sum_{\sac\in\ptnp{\rf}{\ac'}}\ime(\sac)\\
&=\sum_{\ac'\in\ch{\rf(\ac)}}\ime(\ac')\\
&=\sum_{\ac'\in\ch{\rf(\ac)}}\msn(\ac')\ime(\rf(\ac))\\
&=\ime(\rf(\ac))\sum_{\ac'\in\ch{\rf(\ac)}}\msn(\ac')\\
&=\ime(\rf(\ac))\msn(\rf(\ac))\\
&=\ime(\ac)
\end{align*}
as required.

We have now shown that the inductive hypothesis holds for all $\ac\in\rable{\rf}\cap\acs$ and hence:
\begin{align*}
\sum_{\ac\in\ptn{\rf}}\ime(\ac)&=\sum_{\ac\in\ch{\rot}}\sum_{\sac\in\ptnp{\rf}{\ac}}\ime(\sac)\\
&=\sum_{\ac\in\ch{\rot}}\ime(\ac)\\
&=\sum_{\ac\in\ch{\rot}}\msn(\ac)\ime(\rot)\\
&=\sum_{\ac\in\ch{\rot}}\msn(\ac)\\
&=\msn(\rot)\,.
\end{align*}
The result then follows from Lemma \ref{mrnrlem}.
\end{proof}

\begin{lemma}\label{exptrlem}
For all $t\in[T]$ we have:
\be
\expt{\ln(\nrm{t}{\rot})}\leq\lr\ga|\acs|+\frac{9(e-2)\lr^2}{\ep^2}|\acs|-\lr\expt{\lost{t}}\,.
\ee
\end{lemma}

\begin{proof}
By lemmas \ref{spethlem} and \ref{exlnnrmlem} we have:
\begin{align}
\notag\cexpt{\nrm{t}{\rot}}{\polt{t}}&=\sum_{\st\in\rsts}\cprob{\stt{t}=\st}{\polt{t}}\cexpth{\exp\left(-\sum_{\ac\in\sacs{\st}}\vth(\ac)\right)}{\polt{t}}\\
\notag&\leq\sum_{\st\in\rsts}\cprob{\stt{t}=\st}{\polt{t}}-\lr\sum_{\st\in\rsts}\cprob{\stt{t}=\st}{\polt{t}}\loss{\st}{\rft{t}}\\
\notag&~~~~+\lr\ga\sum_{\st\in\rsts}\frac{\ime(\zti{t}(\st))\cprob{\stt{t}=\st}{\polt{t}}}{\cprob{\zti{t}(\st)\in\rable{\stt{t}}}{\polt{t}}}\\
\label{exptrlemeq1}&~~~~+\frac{9(e-2)\lr^2}{\ep^2}\sum_{\st\in\rsts}\sum_{\ac\in\sacs{\st}}\frac{\cprob{\stt{t}=\st}{\polt{t}}}{\cprob{\ac\in\rable{\stt{t}}}{\polt{t}}}\,.
\end{align}
First note that:
\begin{equation}\label{exptrlemeq2}
\sum_{\st\in\rsts}\cprob{\stt{t}=\st}{\polt{t}}=1
\end{equation}
and:
\begin{equation}\label{exptrlemeq3}
\sum_{\st\in\rsts}\cprob{\stt{t}=\st}{\polt{t}}\loss{\st}{\rft{t}}=\cexpt{\lost{t}}{\polt{t}}\,.
\end{equation}
Next note that for all $\st\in\rsts$ we have $\zti{t}(\st)\in\ptn{\rft{t}}$. Note also that for all $\ac\in\ptn{\rft{t}}$ we have $\ac=\zti{t}(\stt{t})$ if and only if $\ac\in\rable{\stt{t}}$. Hence, by Lemma \ref{imestalem}, we have that:
\begin{align}
\notag\sum_{\st\in\rsts}\frac{\ime(\zti{t}(\st))\cprob{\stt{t}=\st}{\polt{t}}}{\cprob{\zti{t}(\st)\in\rable{\stt{t}}}{\polt{t}}}&=\sum_{\ac\in\acs}\sum_{\st\in\rsts}\frac{\indi{\ac=\zti{t}(\st)}\ime(\ac)\cprob{\stt{t}=\st}{\polt{t}}}{\cprob{\ac\in\rable{\stt{t}}}{\polt{t}}}\\
\notag&=\sum_{\ac\in\ptn{\rft{t}}}\frac{\ime(\ac)}{\cprob{\ac\in\rable{\stt{t}}}{\polt{t}}}\sum_{\st\in\rsts}\indi{\ac=\zti{t}(\st)}\cprob{\stt{t}=\st}{\polt{t}}\\
\notag&=\sum_{\ac\in\ptn{\rft{t}}}\frac{\ime(\ac)\cprob{\ac=\zti{t}(\stt{t})}{\polt{t}}}{\cprob{\ac\in\rable{\stt{t}}}{\polt{t}}}\\
\notag&=\sum_{\ac\in\ptn{\rft{t}}}\ime(\ac)\\
\label{exptrlemeq4}&=|\acs|\,.
\end{align}
Finally, note that:
\begin{align}
\notag\sum_{\st\in\rsts}\sum_{\ac\in\sacs{\st}}\frac{\cprob{\stt{t}=\st}{\polt{t}}}{\cprob{\ac\in\rable{\stt{t}}}{\polt{t}}}&=\sum_{\ac\in\acs}\frac{1}{\cprob{\ac\in\rable{\stt{t}}}{\polt{t}}}\sum_{\st\in\rsts}\indi{\ac\in\sacs{\st}}\cprob{\stt{t}=\st}{\polt{t}}\\
\notag&=\sum_{\ac\in\acs}\frac{\cprob{\ac\in\sacs{\stt{t}}}{\polt{t}}}{\cprob{\ac\in\rable{\stt{t}}}{\polt{t}}}\\
\notag&=\sum_{\ac\in\acs}1\\
\label{exptrlemeq5}&=|\acs|\,.
\end{align}
Substituting equations \eqref{exptrlemeq2}, \eqref{exptrlemeq3}, \eqref{exptrlemeq4} and \eqref{exptrlemeq5} into Equation \eqref{exptrlemeq1} gives us:
\be
\cexpt{\nrm{t}{\rot}}{\polt{t}}\leq1-\lr\cexpt{\lost{t}}{\polt{t}}+\lr\ga|\acs|+\frac{9(e-2)\lr^2}{\ep^2}|\acs|\,.
\ee
Since $\ln(\anum)\leq \anum - 1$ for all $\anum>0$, we then have that:
\begin{align*}
\cexpt{\ln(\nrm{t}{\rot})}{\polt{t}}&\leq\cexpt{\nrm{t}{\rot}}{\polt{t}}-1\\
&\leq\lr\ga|\acs|+\frac{9(e-2)\lr^2}{\ep^2}|\acs|-\lr\cexpt{\lost{t}}{\polt{t}}
\end{align*}
which implies the result.
\end{proof}

\begin{lemma}\label{deltolem}
We have:
\be
\delt{1}=-\ln(|\rsts|)\,.
\ee
\end{lemma}

\begin{proof}
Given some $\anod\in\ran{\coms}$ let $\acd{\coms}{\anod}$ be the set of all nodes in $\sacs{\coms}$ that are descendants of $\anod$. We take the inductive hypothesis that for any $\anod\in\ran{\coms}$ we have:
\be
\sum_{\ac\in\acd{\coms}{\anod}}\ln(\polt{1}(\ac))=-\ln(\nsts(\anod))\,.
\ee
To prove this by induction, we may assume that for all $\anod'\in\ch{\coms(\anod)}\cap\lns$ we have:
\be
\sum_{\ac\in\acd{\coms}{\anod'}}\ln(\polt{1}(\ac))=-\ln(\nsts(\anod'))\,.
\ee
Note then, that:
\begin{align*}
\sum_{\ac\in\acd{\coms}{\anod}}\ln(\polt{1}(\ac))&=\ln(\polt{1}(\coms(\anod)))+\sum_{\anod'\in\ch{\coms(\anod)}\cap\lns}\sum_{\ac\in\acd{\coms}{\anod'}}\ln(\polt{1}(\ac))\\
&=\ln\left(\frac{\nsts(\coms(\anod))}{\nsts(\anod)}\right)-\sum_{\anod'\in\ch{\coms(\anod)}\cap\lns}\ln(\nsts(\anod'))\\
&=\ln\left(\frac{\nsts(\coms(\anod))}{\nsts(\anod)}\right)-\ln\left(\prod_{\anod'\in\ch{\coms(\anod)}\cap\lns}\nsts(\anod')\right)\\
&=\ln\left(\frac{\nsts(\coms(\anod))}{\nsts(\anod)}\right)-\ln(\nsts(\coms(\anod)))\\
&=-\ln(\nsts(\anod))
\end{align*}
as required.

We have now shown that the inductive hypothesis holds for all $\anod\in\ran{\coms}$. In particular, we have:
\begin{align*}
\delt{1}&=\sum_{\ac\in\sacs{\coms}}\ln(\polt{1}(\ac))\\
&=\sum_{\ac\in\acd{\coms}{\rot}}\ln(\polt{1}(\ac))\\
&=-\ln(\nsts(\rot))\,.
\end{align*}
The result then follows from Lemma \ref{mrnrlem}.
\end{proof}

\begin{lemma}\label{finallem}
We have:
\be
\expt{\sum_{t\in[T]}\lost{t}}\leq\sum_{t\in[T]}\loss{\coms}{\rft{t}}+2\sqrt{\left(\frac{6\ln(T)}{\ep}+\frac{9(e-2)}{\ep^2}\right)|\acs|\ln(|\rsts|)T}\,.
\ee
\end{lemma}

\begin{proof}
Note that $\delt{T+1}\leq0$. Hence, by lemmas \ref{delmdellem1} and \ref{deltolem}  we have:
\begin{align*}
\frac{\ln(|\rsts|)}{\lr}&=-\frac{\delt{1}}{\lr}\\
&\geq\frac{1}{\lr}\expt{\delt{T+1}-\delt{1}}\\
&=\frac{1}{\lr}\expt{\sum_{t\in[T]}(\delt{t+1}-\delt{t})}\\
&=-\frac{1}{\lr}\sum_{t\in[T]}\expt{\delt{t}-\delt{t+1}}\\
&\geq-\sum_{t\in[T]}\loss{\coms}{\rft{t}}-\frac{1}{\lr}\sum_{t\in[T]}\expt{\ln(\nrm{t}{\rot})}\,.
\end{align*}
Applying Lemma \ref{exptrlem} then gives us:
\begin{align*}
\frac{\ln(|\rsts|)}{\lr}&\geq\sum_{t\in[T]}\expt{\lost{t}}-\sum_{t\in[T]}\loss{\coms}{\rft{t}}-\left(\ga+\frac{9(e-2)\lr}{\ep^2}\right)|\acs|T
\end{align*}
so that:
\be
\expt{\sum_{t\in[T]}\lost{t}}\leq\sum_{t\in[T]}\loss{\coms}{\rft{t}}+\frac{\ln(|\rsts|)}{\lr}+\left(\ga+\frac{9(e-2)\lr}{\ep^2}\right)|\acs|T\,.
\ee
Substituting in the values of $\ga$ and $\lr$ and applying Lemma \ref{mrnrlem} then gives us the result.
\end{proof}

Combining lemmas \ref{diffplem}, \ref{comcomlem}, \ref{eventlem} and \ref{finallem} gives us Theorem \ref{mainth}.


\newpage

\end{document}